\documentclass{llncs}

\usepackage{graphicx}
\usepackage{verbatim}
\usepackage{float}
\usepackage{amssymb}
\usepackage[
backend=biber,
style=numeric,
sorting=none
]{biblatex}
\usepackage{amsmath}

\begin{document}

\title{Tokenized Data Markets}
\author{Bharath Ramsundar \and Roger Chen \and
Alok Vasudev \and Rob Robbins \and Artur Gorokh}

\institute{Computable\\
\email{bharath@computable.io},\\
\email{roger@computable.io},\\ 
\email{alok@computable.io},\\ 
\email{rob@computable.io},\\ 
\email{artur@computable.io}\\ 
}

\maketitle

\begin{abstract}
We formalize the construction of decentralized data markets by introducing the mathematical construction of tokenized data structures, a new form of incentivized data structure. These structures both specialize and extend past work on token curated registries and distributed data structures. They provide a unified model for reasoning about complex data structures assembled by multiple agents with differing incentives. We introduce a number of examples of tokenized data structures and introduce a simple mathematical framework for analyzing their properties. We demonstrate how tokenized data structures can be used to instantiate a decentralized, tokenized data market, and conclude by discussing how such decentralized markets could prove fruitful for the further development of machine learning and AI.
\end{abstract}

\section{Introduction}

Data markets connect buyers and sellers of datasets with one another. Such markets may prove a fundamental new primitive for the next stage of the internet, especially as machine learning and AI systems continue to embed themselves at the heart of the modern technology ecosystem. Learning methods are often data hungry, and require access to large datasets in order to make accurate predictions. Unfortunately, such datasets are nontrivial to gather, and existing data markets lack liquidity. Only the largest and most connected organizations have the resources to secure access to the data they require. The construction of liquid data markets would fundamentally shift this distribution of power and facilitate the broad adoption of machine learning methods.

How can such a data market be constructed? One option is to identify a trusted entity to act as a centralized data broker. Such a broker could enable transactions between buyers and sellers of data by storing datasets on-site and transferring them upon payment. Unfortunately, such a model creates a heavy burden of trust; how can buyers and sellers know that the broker is behaving fairly? Centralized cryptocurrency exchanges already have a checkered history of fraud and theft. It seems all too likely a centralized data exchange could fall prey to similar problems. For these reasons, the construction of a \textit{decentralized} data exchange could prove an enabling technology for liquid data markets. Such an exchange would facilitate transactions of data between buyers and sellers without the need for a trusted third-party broker. Furthermore, tokenization of data offers a powerful new primitive for solving cold-start problems that generally make boostrapping a marketplace difficult. While many might agree that pooling data creates non-zero sum value for all participants, most hesitate to be the first to contribute without some contractual guarantee of value. With decentralized data markets, the earliest contributors see financial incentive because they can receive tangible cryptoeconomic assets (tokens) even before buyers enter the market.

The construction of a decentralized data exchange is not straightforward. How can participants ensure that their datasets are stored and transferred correctly? How can cheaters be caught and removed from the system? These are deep questions which delve into the heart of multiparty protocols. Luckily, the advent of blockchain based systems with associated smart contract platforms \cite{buterin2013ethereum} has triggered significant research into the design of multi-agent systems designed to perform nontrivial work. For example, prediction markets \cite{peterson2015augur}, decentralized token exchanges \cite{warren20170x}, curation markets \cite{delarouviere2017curationmarkets}, token curated registries \cite{goldin2017tcr}, storage markets \cite{wilkinson2014storj}, and computational markets \cite{teutsch2017scalable} provide various examples of systems designed to perform useful work by coordinating selfish actors. Primitives introduced by such protocols can be repurposed to serve as a foundation for decentralized data markets.

The token curated registry (TCR) \cite{goldin2017tcr} in particular provides a powerful abstraction for how a collection of participants can work together to build a curated list. For example, such a list could contain the names of colleges which enable students to rapidly pay back student debt after graduation. Basic implementations of TCRs in Solidity already exist \cite{goldinTcrImpl}. However these implementations have a number of limitations. For example, storage is typically on-chain for simplicity; this basic design wouldn't permit for the construction of a list of images since images are too large to be stored on existing smart contract platforms. In addition, the contents of the registry are publicly visible, so sensitive information can't be assembled. 

To overcome these issues, it proves useful to specialize the basic design of token curated registries to fit within a structured framework which explicitly allows for off-chain storage and private data. In addition, we introduce the new notion of recursively nesting TCRs to allow for the construction of more complex data structures. We call this modified mathematical class of structures \textit{tokenized data structures}. Tokenized data structures allow for a number of improvements over simple on-chain TCR implementations:

\begin{itemize}
\item \textbf{Off-chain storage}: At present, simple token curated registries cannot hold large datasets since the registry contents are stored on-chain. Tokenized data structures on the other hand allow for the storage of data elements which may be too large to fit on-chain. Such data elements could be stored on IPFS \cite{benet2014ipfs} or similar storage networks. Enabling off-chain storage significantly extends the types of data structures that can be constructed. A decentralized data exchange could store all its datasets off-chain in this fashion. Alternatively, a tokenized map could be constructed to provide an alternative to Google maps. Note that such a data exchange or tokenized map would require the storage of terabytes and perhaps petabytes of data. Coordination mechanisms that enable a tokenized data structure to effectively access distributed off-chain state will prove fundamental for these applications. We discuss such mechanisms later in this work.
\item \textbf{Private Data}: Decentralized data exchanges will require that only the rightful owners of datasets be able to access data. For this reason, data indexed in the tokenized data structure must be kept private. Similarly, the tokenized map introduced above could have regions of the map restricted to the general public (say for military bases), or the map could cover private property; a token curated Disneyland map may require a payment to Disney in order to access. Tokenized data structures need to allow private data to be maintained as part of its structure. Agents who wish to access data must purchase \textit{membership} in order to access such data. In a decentralized data exchange, buyers of data must purchase membership in the data in order to access.
\item \textbf{Recursive Nesting}: Some tokenized data structures could require significant capital expenditure to construct. For the case of a map dataset, it’s possible that mapping a new city might require a mapper to expend capital gathering the mapping information needed to add a new entry to a tokenized map with existing token $\mathcal{T}$. Let's suppose that our mapper lacks the needed funds, but has an entrepreneurial mindset. For this purpose, the mapper can construct a new city token $\mathcal{CT}$ which she can use to fund her data gathering efforts. This token is tied to the broader map token $\mathcal{T}$ so that our mapper doesn't need to exit the existing mapping ecosystem. The mapper can sell a fraction of her founder $\mathcal{CT}$ tokens to obtain the funds necessary to gather the first maps for the new city. In order to attract investors to $\mathcal{CT}$, there must be mechanisms by which $\mathcal{CT}$ token holders can obtain rights to future monetary returns from the new city map. We introduce mathematical structures, namely a membership model, that provide these returns. 
\end{itemize}

We start by reviewing the literature for related ideas, then proceed to provide a number of practical examples of tokenized data structures, culminating with the construction of a decentralized data market via a tokenized data structure. We then use these examples to motivate a mathematical framework for analyzing tokenized data structures, and prove some basic economic theorems governing their behavior. We discuss how decentralized data markets may enable the advancement of machine learning and AI, and conclude by highlighting a few open problems relating to tokenized data structures.

\section{Related Work}

Bitcoin \cite{nakamoto2008bitcoin} introduced the first broadly adopted token incentivized scheme. Its proof of work mining algorithm provided an incentive for miners to run large computations in return for token rewards. Despite its impact, Bitcoin does not provide an easy way for developers to build applications on top of the core protocol. Ethereum \cite{buterin2013ethereum} extends the Bitcoin design with a (quasi) Turing complete virtual machine on top \cite{wood2014ethereum} capable of executing smart contracts. A number of smart contract systems have been devised which implement powerful incentive systems such as prediction markets \cite{peterson2015augur}, decentralized exchanges \cite{warren20170x}, and computational markets \cite{teutsch2017scalable}.

Both Bitcoin and Ethereum were originally designed to use proof-of-work (PoW) mining. In such systems, teams of miners compete for the right to propose the next cleared set of transactions by solving computational challenge problems (typically hash inversion). PoW has proven a robust and powerful security mechanism, but at the cost of tremendous electricity and resource consumption. For this reason, a parallel line of work has investigated proof-of-stake (PoS) mining algorithms \cite{king2012ppcoin, kiayias2017ouroboros, buterin2017casper}. Such algorithms require that miners hold "stake" in the form of coins held in the economic system. Miners are selected to propose the next cleared "block" of transactions according to their stake. To keep miners honest, a number of "slashing conditions" \cite{buterin2017casper} have been proposed which punish dishonest miners. Although proof-of-stake was originally envisioned as a scheme for securing blockchains, it has become clear that computational stake serves as a powerful scheme to coordinate agents to perform useful work. Many protocols \cite{peterson2015augur, teutsch2017scalable, goldin2017tcr, goldin2018tcr11} rely upon staking mechanisms to coordinate actors to perform useful work and upon slashing conditions to punish dishonest behavior.

Token curated registries \cite{goldin2017tcr} (TCRs) in particular allow for the construction of lists that are maintained by a set of curators. These curators must be bonded into the TCR by placing tokens at stake. The bonding of curators creates natural incentive structures that help the listing take natural form. The original TCR design was subsequently modified to add "slashing" conditions that punish token holders who don't participate in votes regularly. \cite{goldin2018tcr11}. A number of related designs to TCRs such as curation markets for coordinating agents around shared goals \cite{delarouviere2017curation, delarouviere2017curationmarkets} have also been proposed. Refinements such as bonding curves \cite{delarouviere2017tokens2} have been proposed which allow for additional flexibility in the choice of how participants are rewarded with tokens for their efforts.

It's important to note however that unlike PoW algorithms, PoS methods have not been tested yet with large real world deployments. A line of recent work has demonstrated that long-range attacks \cite{gazi2018stake}, where miners wait until they can remove stake from the system to launch attacks, may seriously compromise the security of such systems. Nevertheless, the flexibility and energy friendliness of PoS systems means that research into the design of systems continues full steam.

A different line of work has investigated distributed hash tables \cite{stoica2003chord, kaashoek2003koorde, maymounkov2002kademlia}, data structures which enable decentralized networks of participants to maintain useful information. Such decentralized data structures form the foundations of modern internet architecture and also feature prominently in the design of many tokenized protocols \cite{wood2014ethereum, wilkinson2014storj}. One way of contextualizing tokenized data structures would be to view them as the blending of ideas from PoS incentive schemes with distributed hash table style decentralized storage. Protocols such as Storj \cite{wilkinson2014storj} and Filecoin \cite{filecoin2017} have explored this design space. Storj proposes a peer-to-peer storage network where availability of data is guaranteed by a challenge response scheme and where storage nodes are rewarded with tokens. The locations of shards of data are stored on an underlying Kademlia distributed hash table \cite{maymounkov2002kademlia}.

Unlike systems such as Storj, tokenized data structures introduce the notion of recursive sub-tokens enabling different agents to construct parts the tokenized data structure. These sub-tokens draw from past work on non-fungible tokens \cite{eip721}, which create custom tokens tied to particular physical or virtual entities. For example, Cryptokitties \cite{cryptokitties2018} associates separate non-fungible tokens to instances of collectible virtual cats (the aforementioned "Cryptokitties").

Tokenized data structures also draw some inspiration from past work on decentralized cryptocurrency exchanges \cite{warren20170x}. However, the needs for a decentralized data exchange to secure large off-chain datasets means that it's not feasible to directly adopt decentralized exchange protocols for data transactions.

\section{Examples of tokenized data structures}

Before introducing a formal mathematical definition of tokenized data structures, it will be useful to discuss a number of different types of tokenized data structures to build intuition. We present a series of tokenized data structures of increasing complexity, culminating in the construction of a decentralized data market. An important design theme that will emerge in this discussion is the recursive nature of tokenized data structures, which means that such structures can be fruitfully combined to build more complicated systems.

\subsection{Distributed Hash Table}

A tokenized data structure with no associated token but with off-chain storage forms a distributed hash table. Assuming that the tokenized data structure is implemented on a smart-contract platform, the lookup table mapping keys to data locations can be implemented as a smart contract data structure stored on-chain as illustrated in Figure~\ref{fig:dht}.

\begin{figure*}
  \centering
  \includegraphics[width=0.75\textwidth]{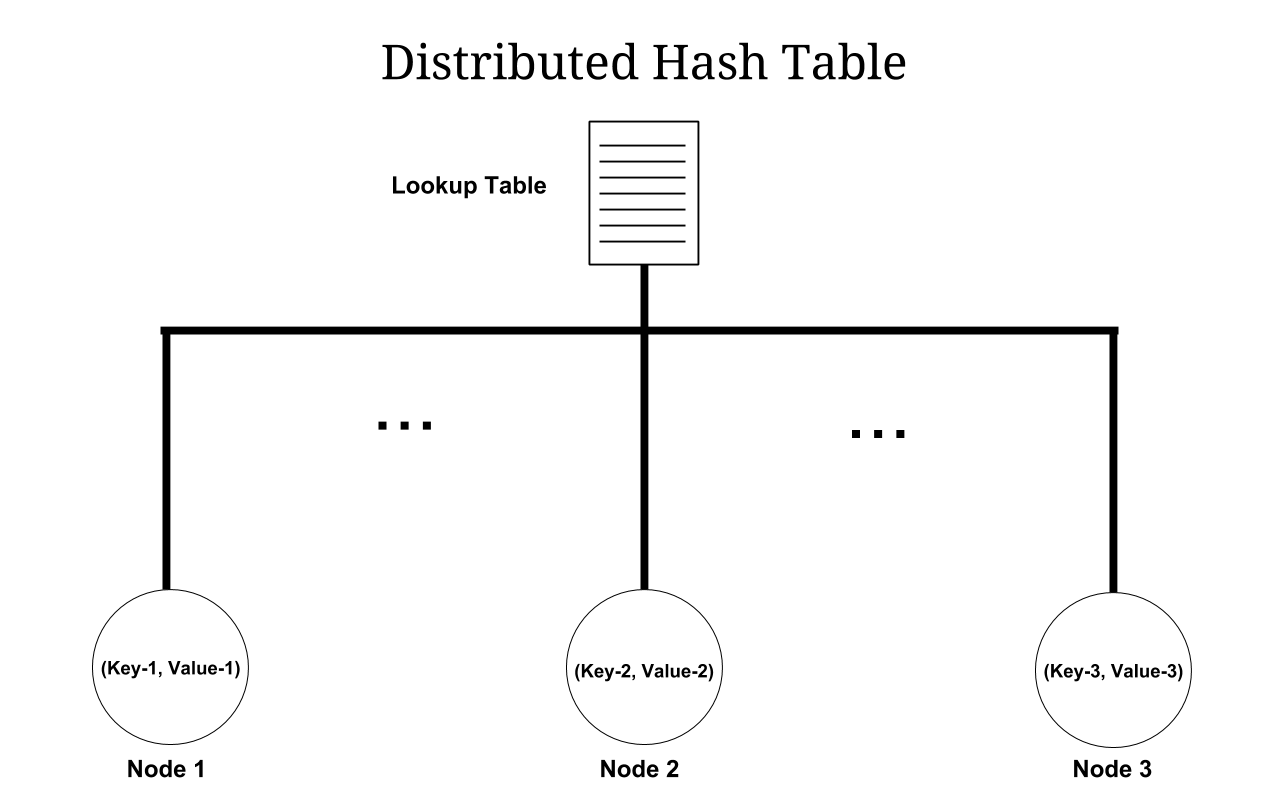}
  \caption{A distributed hash table is a tokenized data structure (with no associated token).}
  \label{fig:dht}
\end{figure*}

\subsection{Token Curated Registry}

A simple token curated registry is a special case of a tokenized data structure with no off-chain storage and no private data (Figure~\ref{fig:thr}). Note that the concept of a token curated registry is often disused quite generally, so it would be equally fair to argue that all tokenized data structures are themselves special cases of token curated registries.

\begin{figure*}
  \centering
  \includegraphics[width=0.75\textwidth]{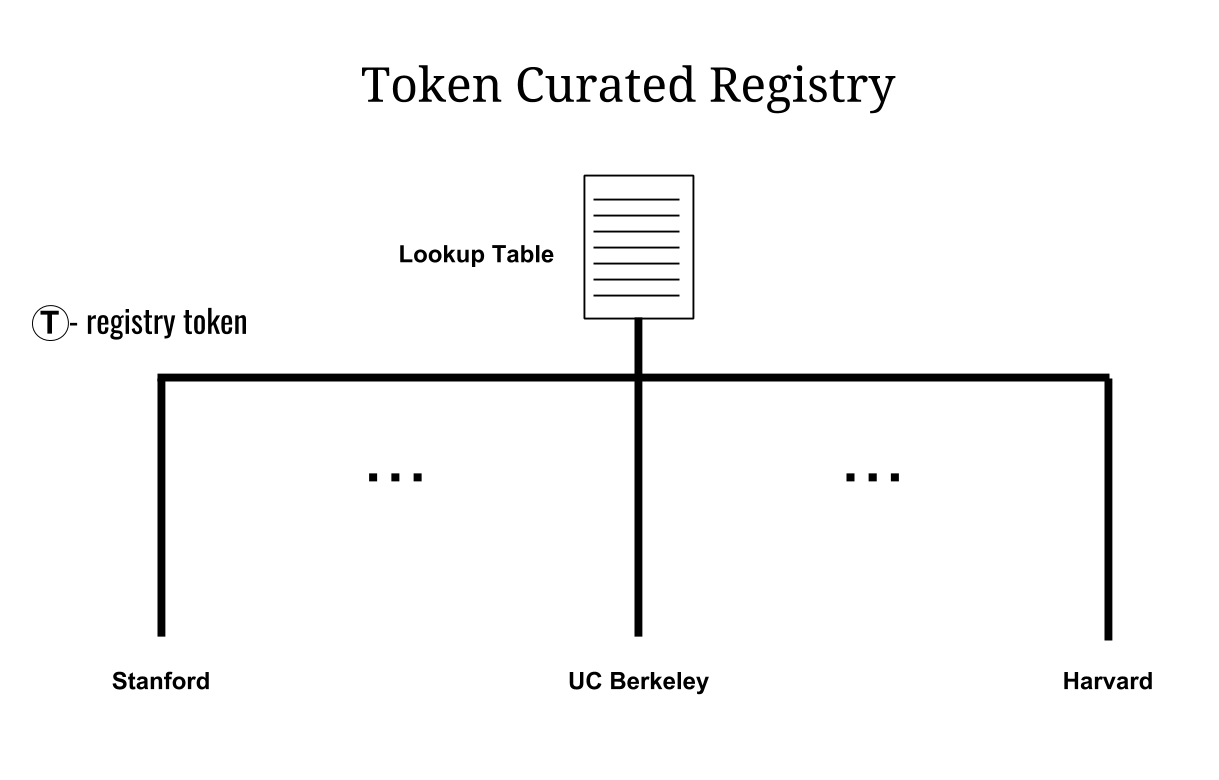}
  \caption{A simple token curated registry is a special case of a tokenized data structure with on-chain storage and public visibility.}
  \label{fig:thr}
\end{figure*}

\subsection{Tokenized Dataset}

A tokenized dataset is a distributed hash table that has an associated token $\mathcal{T}$. Alternatively, the tokenized dataset can be viewed as token curated registry but with the addition of off-chain storage. Figure~\ref{fig:tcdataset} illustrates a token curated image dataset with public data visibility while Figure~\ref{fig:privtcdataset} illustrates a token curated image dataset with private data visibility.

\begin{figure*}[h]
  \centering
  \includegraphics[width=0.75\textwidth]{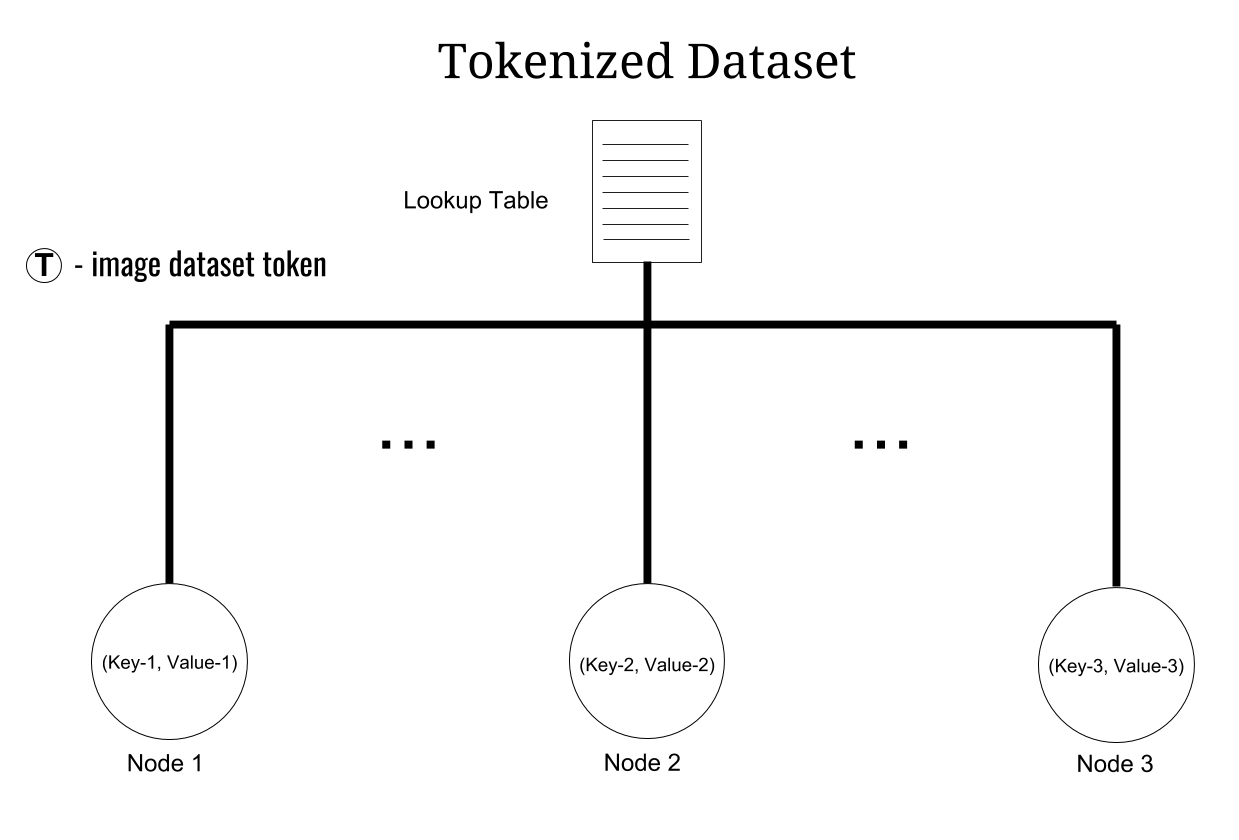}
  \caption{A tokenized dataset. For example, the tokenized dataset might hold an image dataset with the images stored off-chain. Note that the data is set to be publicly visible.}
  \label{fig:tcdataset}
\end{figure*}

\begin{figure*}[h]
  \centering
  \includegraphics[width=0.75\textwidth]{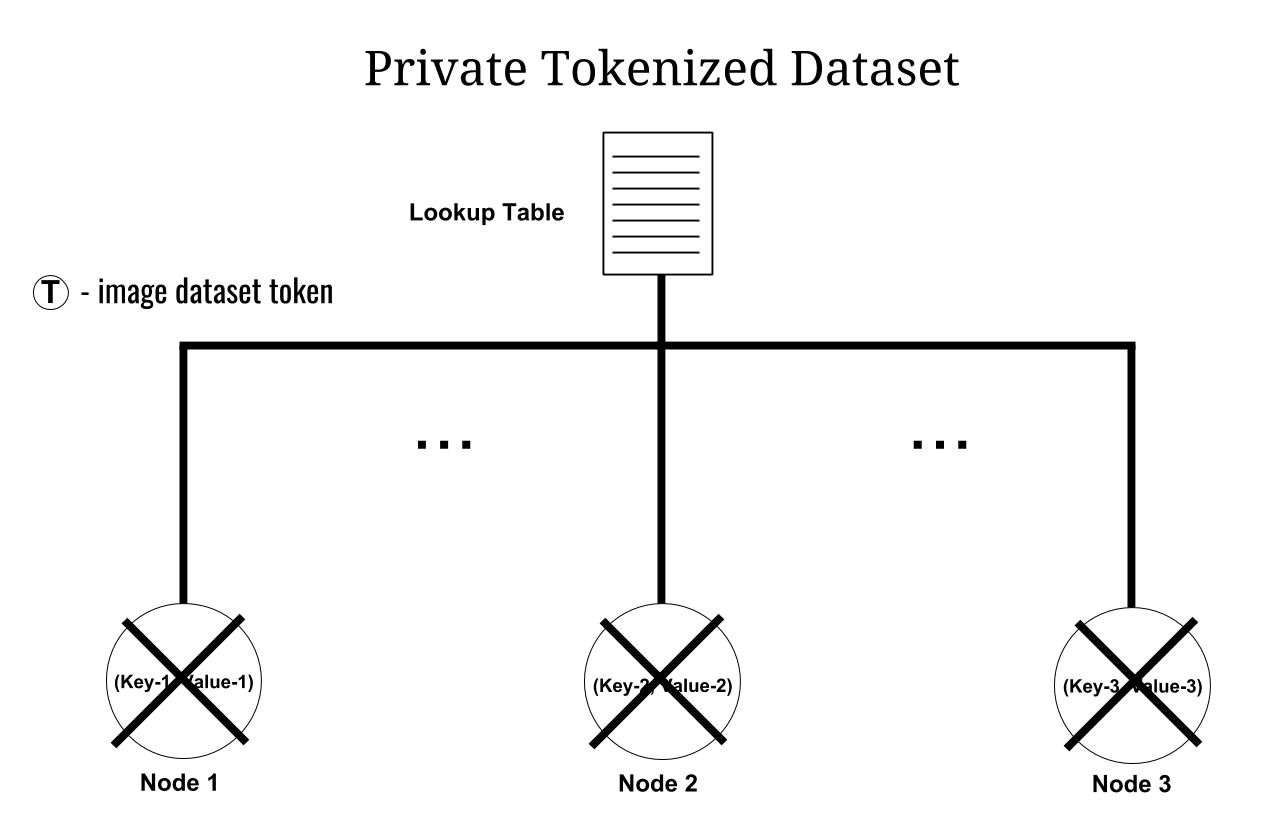}
  \caption{A Private tokenized dataset. In this particular representation, data stored off-chain are not publicly visible.}
  \label{fig:privtcdataset}
\end{figure*}

It's illustrative to imagine how a large image dataset like ImageNet \cite{deng2009imagenet} could have been gathered with a private tokenized dataset rather than through Amazon's Mechanical Turk. Workers who contributed images would be rewarded by being issued tokens of type $\mathcal{T}$ that could be renumerated at a future date for currency.

In the longer run, tokenized datasets may prove to be a far more powerful tool for incentivizing the construction of large datasets than Mechanical Turk. Unlike Mechanical Turk, tokenized datasets have support for recursive sub-tokens which allows workers to be rewarded with a share of future financial rewards from the dataset. This expectation of future rewards is a powerful economic driver. The modern startup functions because founding employees accept severe risks in expectation of future rewards from their fractional ownership of the company. Similarly, tokenized datasets may enable "data startups" which work collaboratively to construct datasets of significantly greater scale and utility than ImageNet.

\subsection{Tokenized Map}

A tokenized map (Figure~\ref{fig:tcm}) is a two dimensional grid with off-chain storage for local information at grid points. A tokenized map could be used to incentivize the construction a version of Google Maps. Local businesses could pay for transactions to add their business information to the tokenized map.

\begin{figure*}
  \centering
  \includegraphics[width=0.75\textwidth]{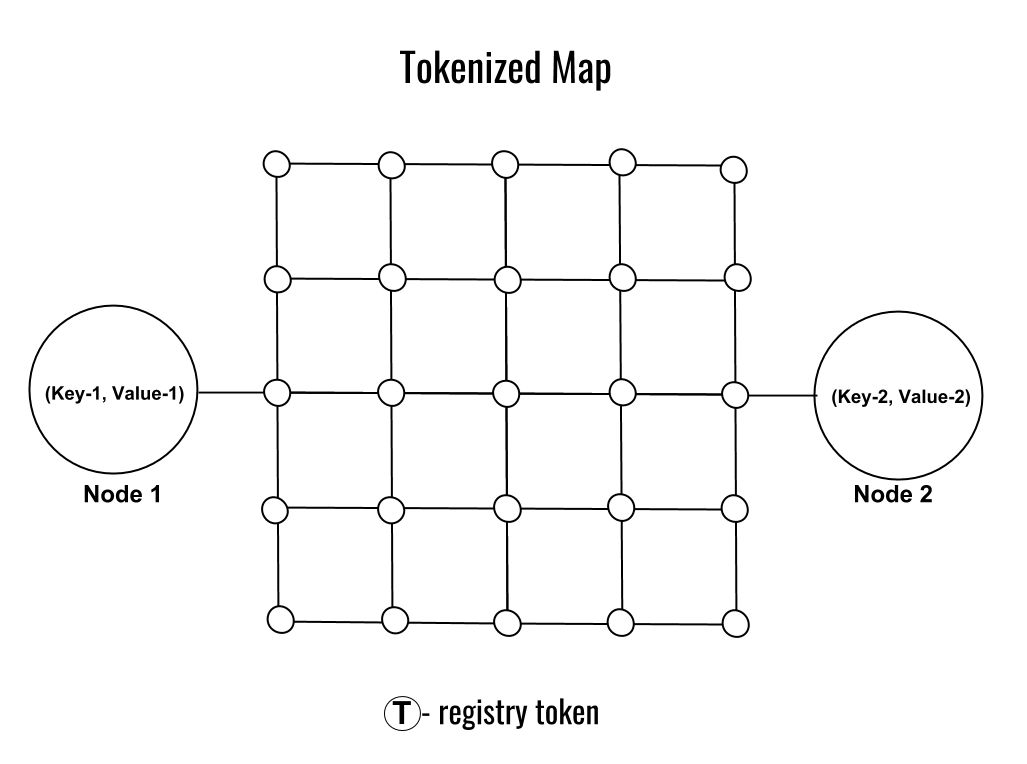}
  \caption{A tokenized map.}
  \label{fig:tcm}
  \caption{A tokenized map could be used to construct a dataset analogous to that owned by Google Maps. Business owners could pay to have their information added to the map.}
\end{figure*}

\subsection{Tokenized Tree}

Let's suppose that adding elements to a tokenized data structure would take significant capital outlay. For example, the tokenized data structure we wish to construct might be a vast phylogenetic tree that holds all the world’s genomic data (Figure~\ref{fig:tctree}). In this case, adding one new individual to the phylogenetic tree could take a sizable sum of money to pay for the needed genetic sequencing. Since there may not exist interested individuals who are willing to directly pay for the construction of this tree, the tokenized tree suffers from a severe cold start problem.

More generally, if there exists a substructure in a tokenized data structure that is difficult to construct, a token can be constructed that is tied to this substructure. For example, let’s suppose that the South Asian branch of the phylogenetic tree is sparse. An interested network participant can contribute genetic material in return for \textit{SouthAsianBranchTokens} (SABTs). If a future agent pays to access the private data on the South Asian branch, payments will be made to SABT holders. The anticipation of these future payments serve as an incentive to encourage contribution of data elements to the South Asian branch. Note that the portions of the South Asian branch must be kept private else there will be no incentive to pay for data access. Conceptually, the SABT holders have a form of ownership in the South Asian branch of the tokenized phylogenetic tree.

Similarly, a \textit{PolynesianBranchTokens} (PBTs) may incentivize gathering of genomic data for the Polynesian branch of the tokenized phylogenetic tree. But, it's important to note that entirely different organizations may be involved with this branch of the tree! That is, SABTs and PBTs may be used by different organizations, with their efforts coordinated by the decentralized tokenized phylogenetic tree. This potential for decentralized coordination of disparate organizations could enable complex datasets to be assembled.

\begin{figure*}
  \centering
  \includegraphics[width=0.75\textwidth]{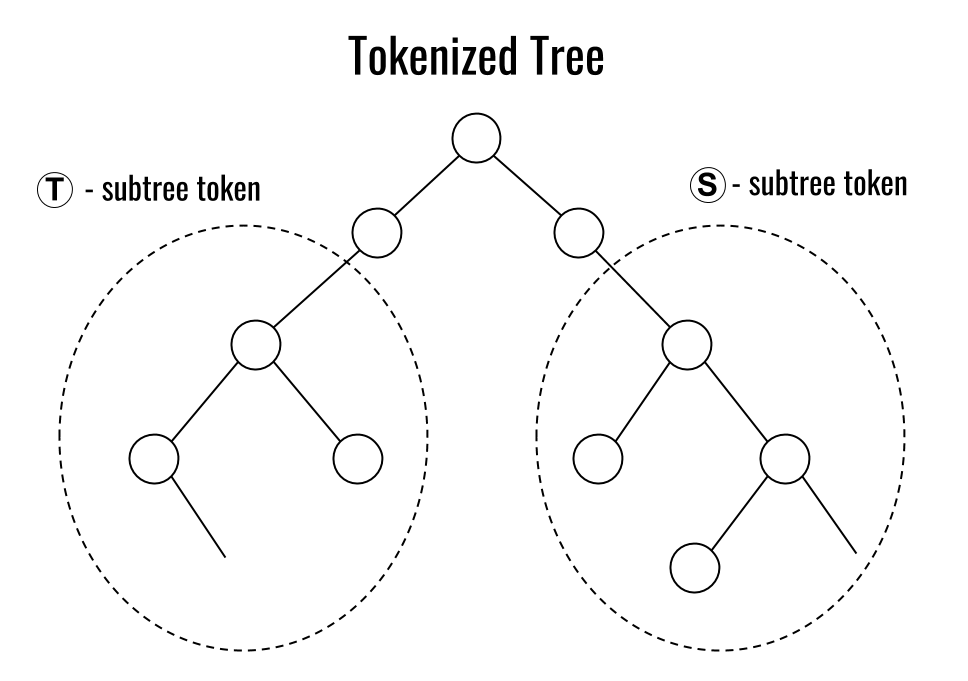}
  \caption{A tokenized tree. Different subtrees are incentivized by different tokens.}
  \label{fig:tctree}
\end{figure*}

\subsection{Decentralized Data Markets}

A decentralized data market would provide data liquidity by enabling data transactions between buyers and sellers of data. How can such a market be instantiated as a tokenized data structure? Luckily, we've already discussed many of the compontent pieces of such a market structure already. Individual datasets can be stored on the market as (private) tokenized datasets. The collection of such tokenized datasets can be organized itself as a simple token curated registry. Put another way, a tokenized data market is defined as a simple token curated registry of (private) tokenized datasets. Figure~\ref{fig:tcdatreg} illustrates a tokenized data market.

A tokenized data market could be used to construct a decentralized data exchange where participants can access various useful types of data by accessing constituent tokenized datasets. Agents would be incentivied to construct new datasets in anticipation of future rewards for token holdings in such datasets via dataset tokens.

\begin{figure*}
  \centering
  \includegraphics[width=\textwidth]{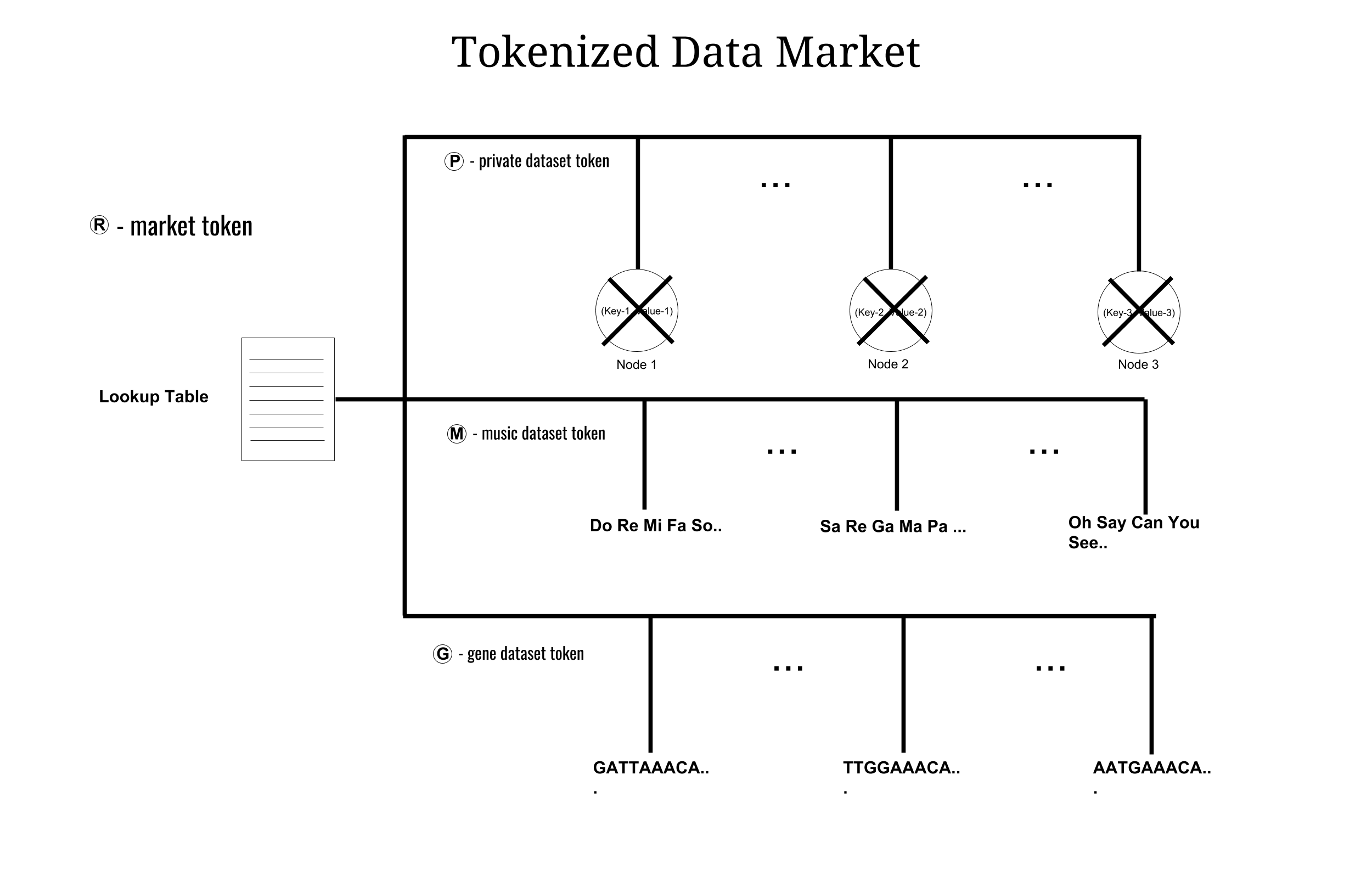}
  \caption{A tokenized data market. Different datasets are incentivized by different tokens.}
  \label{fig:tcdatreg}
\end{figure*}

\section{Mathematical Definitions}

In this section, we provide formal mathematical definitions of tokenized data structures and analyze a number of their mathematical properties. 

\subsection{Formal Definitions}
A tokenized data structure $\mathcal{TD}$ is a collection of elements $e_1,\dotsc, e_N$ with an optional token type $\mathcal{T}$, associated ledger $\mathcal{L}$ that maps agents $a_j$ to their token holdings $t_j$, and metadata $\mathcal{M}$ that annotates elements $e_i$ with additional information. Formally, we can write $\mathcal{TD}$ as a tuple $((e_1,\dotsc,e_N), \mathcal{T}, \mathcal{L}, \mathcal{M})$ (or simply $(e_1,\dotsc, e_N)$ in the token-free, metadata-free case).

Recursively, each element $e_i$ in this may itself be a tokenized data structure $\mathcal{TD}_i = ((e_{i1}, \dotsc, e_{iM}), \mathcal{T}_i, \mathcal{L}_i, \mathcal{M}_i)$.  Alternatively, $e_i$ may be a terminal leaf node containing associated data $v_i$. Note that $v_i$ may be stored on-chain or off-chain, depending on the capacities of the system on which the tokenized data structure is implemented.

For many proofs, it will be useful to talk about the economic value of a particular token. However, such discussions depend on the choice of base currency. Following conventions from the literature \cite{gorbunov2015democoin}, we adopt the notation $\#A$ to denote $A$ units of monetary value.

\begin{definition}[Economy Size]
Let $\mathcal{T}$ denote a token type. Let $n$ denote the number of such tokens tracked in ledger $\mathcal{L}$. Let $\#a$ denote the monetary value of one such token in the base currency and let $\#(\mathcal{T})$ denote the number of tokens of type $\mathcal{T}$ available. Then the size of the token economy $\text{size}(t)$ is defined to be $a\#(\mathcal{T})$ units.
\end{definition}

\subsection{Operations and Parameters}

This section introduces the operations that can be performed on a tokenized data structure and the associated parameters that control the specific behavior of a tokenized data structure under these operations.

There are four classes of operations supported by a tokenized data structure:\textit{candidacy}, \textit{challenges}, \textit{forks}, and \textit{queries}. Candidacy is the process by which new elements are proposed for addition to the tokenized data structure $\mathcal{TD}$. Challenges allow for the legitimacy of elements of $\mathcal{TD}$ to be formally challenged. Forks split $\mathcal{TD}$ into two parts. Queries allow for private elements in the tokenized data structure to be viewed. In the remainder of this section, we expand on these brief definitions and discuss the parameters that govern each operation. Table~\ref{tab:param} summarizes all operations and parameters.

\begin{table}[h!]
\begin{center}
\begin{tabular}{ |c|c|c|c| } 
 \hline
 Operation & Parameter & Symbol & Default Value  \\ 
 \hline
 \hline
 Candidacy & Deposit & $t_\text{candidateDeposit}$ & -  \\
 \hline
 Candidacy & Voting Period & $t_\text{candidateVote}$ & 3 days  \\
 \hline
 Candidacy & Reward  & $S_\text{candidateReward}$ & 1 token \\
 \hline
 Candidacy & Vote Quorum & $S_\text{candidateQuorum}$ & 66.67\%  \\
 \hline
 Candidacy & Reward Stake Period & $t_\text{candidateStake}$ & -  \\
 \hline
 \hline
 Challenge & Deposit & $S_\text{challengeDeposit}$ & -  \\
 \hline
 Challenge & Voting Period & $t_\text{challengeVote}$ & 5 days  \\
 \hline
 Challenge & Reward & $S_\text{challengeReward}$ & 0  \\
 \hline
 Challenge & Vote Quorum & $S_\text{challengeQuorum}$ & 66.67\%  \\
 \hline
 \hline
 Fork & Deposit & $S_\text{forkDeposit}$ & -  \\
 \hline
 Fork & Voting Period & $t_\text{forkVote}$ & 30 days  \\
 \hline
 Fork & Threshold & $S_\text{forkThreshold}$ & 50\%  \\
 \hline
 \hline
 Query & Deposit & $S_\text{query}$ & - \\
 \hline
 \hline
 - & Offchain & $b_\text{offchain}$ & -  \\
 \hline
\end{tabular}
\end{center}
\caption{This table lists the set of parameters that control a given tokenized data structure $\mathcal{TD}$. Note that the three major classes of $\mathcal{TD}$ actions (candidacy, challenge, and fork) are all represented. }
\label{tab:param}
\end{table}

\subsubsection{Candidacy}

Candidacy is the process by which new elements are proposed for addition to a tokenized data structure.

\paragraph{Candidacy Deposit}
The candidate deposit $S_\text{candidateDeposit}$ controls the amount that a token holder must stake in order to propose the addition of an element to the tokenized dataset.

Note that for many practical applications, $S_\text{candidateDeposit}$ may be set to 0 in order to lower barriers for potential candidates to participate in the construction of the tokenized dataset. The danger of setting $S_\text{candidateDeposit} = 0$ is of course that spamming the market becomes much easier.

\paragraph{Candidacy Voting Period}
The candidate voting period $t_\text{candidateVote}$ controls the amount of time that token holders have to vote on a new data candidate.

\paragraph{Candidacy Reward}
The reward $S_\text{candidateReward}$ issued to a candidate for an accepted addition to the tokenized dataset.

\paragraph{Candidacy Vote Quorum}
The percentage of token holders $S_{candidateQuorum}$ who must vote to authorize the addition of a new candidate to the $\mathcal{TD}$.

\subsubsection{Challenge}
Challenges are the mechanism by which token holders can dispute the suitability of a given element for membership in the tokenized data structure. The challenge mechanism allows token holders to remove data structure elements which no longer add value to the global structure.

\paragraph{Challenge Deposit}
$S_\text{challengeDeposit}$ is the amount that a token holder must stake in order to issue a challenge to a particular element in the tokenized dataset.

\paragraph{Challenge Voting Period}
The challenge voting period $t_\text{challengeVote}$ controls the amount of time that a challenge for a particular deposit is open for token holders to vote upon.

\paragraph{Challenge Reward}

The reward $S_\text{challengeReward}$ issued to successful challengers. It is probably appropriate to set $S_\text{challengeReward}$ equal to $0$ since token holders should be incentivized to remove bad entries. However, it might also be reasonable to set $S_\text{challengeReward} = S_\text{candidateReward}$ in which case seized reward $r$ associated with the element $e$ in question is awarded to the challenger.

\paragraph{Challenge Vote Quorum}
The percentage of token holders $S_{challengeQuorum}$ who must vote to authorize the removal of an element from $\mathcal{TD}$.

\subsubsection{Forks}

Forking is the operation by which one tokenized data structure ca nbe split into two tokenized data structures. All of the token holders in the forked structure must pick one of the two structures as legitimate.

\paragraph{Fork Deposit}

The fork deposit $S_\text{forkDeposit}$ controls the amount of stake that must be placed to request a fork of $\mathcal{TD}$.

\paragraph{Fork Voting Period}

The fork voting period $t_\text{forkVote}$ is the amount of time token holders can vote on a proposed fork.

\paragraph{Fork Threshold}

The fork threshold $S_\text{forkThreshold}$ is the amount of votes that must be placed in favor a forking operation to trigger a fork.

\paragraph{Offchain}
If the boolean value $b_\text{offchain}$ is true, then the tokenized data structure has leaf nodes which store information off-chain. Tokenized datasets and tokenized data registries rely fundamentally on off-chain storage for example.

\subsubsection{Query}

Querying is the operation by which stake holders in a tokenized data structure can request to query private data held in leaf nodes of the structure. It's possible to think of a querying operation as a sort of limited challenge operation.

\paragraph{Query Deposit}
Stake $S_\text{query}$ is the amount of stake required to be able to query private data points stored in leaf nodes of the tokenized data structure.

\subsection{Token Issuance Schedule}

The creators of a tokenized data structure $\mathcal{TD}$ have broad flexibility to control token ownership, supply, and issuance. For example, the token economy for $\mathcal{TD}$ could have a fixed supply, be inflationary, or even deflationary depending on the needs of the particular application at hand. In this section, we briefly discuss some potential token allocation strategies.

\subsubsection{Predetermined Allocation}

The creators of $\mathcal{TD}$ could elect to split all tokens amongst themselves in some agreed upon fashion proportional to their expected work contribution. In this case, $\#(\mathcal{T})$ is fixed and does not change over time.

\subsubsection{Mining}

Tokens can be issued in an on-going fashion to contributors of new elements to $\mathcal{TD}$. (The act of contributing a quasi-finalized candidate to $\mathcal{TD}$ is deemed mining.) To enable mining rewards, the creators of $\mathcal{TD}$ need to set $S_\text{candidateReward} > 0$. In this case, $\#(\mathcal{T})$ grows with time.


\subsection{Network Participants}

In this section, we introduce various agents who participate in the construction of a tokenized data structure $\mathcal{TD}$ and the operations they can perform. Table~\ref{tab:agents} lists the three classes of agents: token holders, makers, and queriers. Figure~\ref{fig:tcdiagram} provides a diagrammatic representation of how a tokenized data structure is constructed by agents in the network. Note that the same entity can play multiple roles. W

\begin{table}[h!]
\begin{center}
\begin{tabular}{ |c|c| } 
 \hline
 \textbf{Participant} & \textbf{Properties} \\ 
 \hline
 Token Holder & Agent $a_j$ with token holdings $\mathcal{L}(a_j) = t_j > 0$ \\ 
 \hline
 Maker & Agent $a_j$ with $\mathcal{L}(a_j) > S_{\text{candidate}}$ can submit candidate datapoint $e$ to $\mathcal{TD}$ \\ 
 \hline
 Querier & Agent $a_j$ with $\mathcal{L}(a_j) > S_{\text{query}}$ can query private data \\
 \hline
\end{tabular}
\end{center}
\caption{Agents participating in construction of tokenized data structure $\mathcal{TD}$. Note that makers and queriers are also token holders.}
\label{tab:agents}
\end{table}

\begin{figure*}
  \centering
  \includegraphics[width=\textwidth]{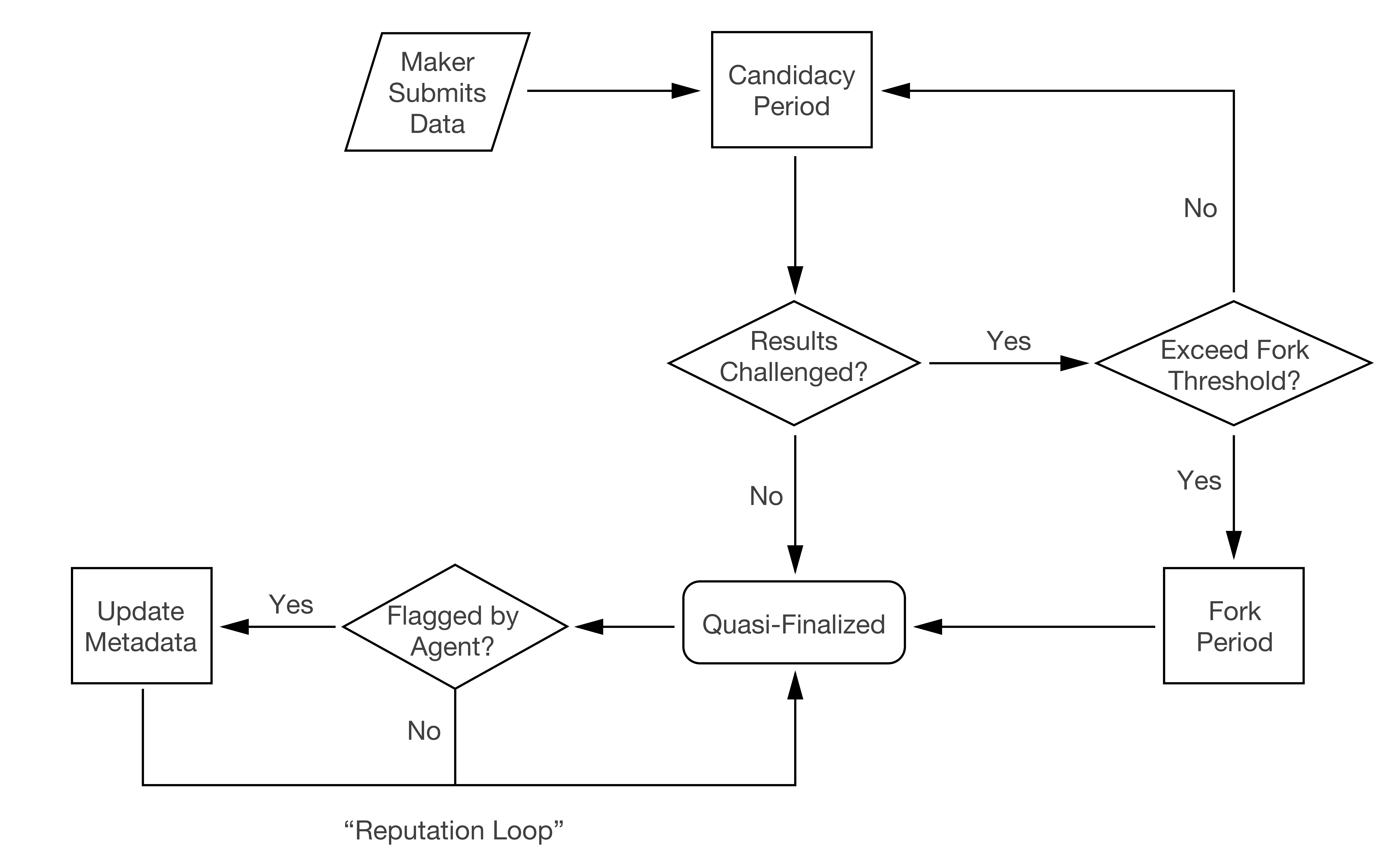}
  \caption{The workflow for agents participating in the construction of a tokenized data structure.}
  \label{fig:tcdiagram}
\end{figure*}

We now formally define each participant in the economy.

\begin{definition}[Token Holders]
A token holder is an agent $a_j$ which holds a nonzero number of tokens $\mathcal{L}(a_j) = t_j > 0$. The token holder's belief that token $t_j$ holds value (say $\#A$ units of value in a base currency) is the primary economic driver responsible for its participation in the construction of $\mathcal{TD}$.
\end{definition}

\begin{definition}[Maker]
A token holder who possesses token holdings $\mathcal{L}(a_j)$ in excess of a set minimum stake $S_\text{candidate}$ can propose that a new candidate element $e$ should be added to $\mathcal{TD}$. Such a modification to $\mathcal{TD}$ must be approved by a quorum $S_\text{candidateQuorum}$ of token holders during a voting period $t_\text{candidateVote}$. For example $S_\text{candidateQuorum} = \frac{2}{3}\text{size}(\mathcal{T})$ would set a quorum at $\frac{2}{3}$ of the token economy size.
\end{definition}

\begin{definition}[Querier]
The querier is any party who is interested in accessing the information stored in $\mathcal{TD}$. The querier must be prepared to pay to renumerate the token holders who have put forward the effort needed to curate the $\mathcal{TD}$.
\end{definition}

With these definitions in place, we can formalize the actions these agents may take.

\begin{definition}[Candidacy]
Let $\mathcal{TD} = \left ((e_1,\dotsc, e_N), \mathcal{T}, \mathcal{L}, \mathcal{M} \right )$ be a tokenized data structure. In a \textit{candidacy}, agent $a$ proposes the addition of element $e$ to $\mathcal{TD}$. Agent $a$ must place $S_{\text{candidate}}$ tokens at stake for the duration of the vote $t_\text{vote}$. If a quorum $S_{\text{quorum}}$ of the token holders recorded in $\mathcal{L}$ authorize the modification, element $e$ is added to the $\mathcal{TD}$ and agent $a$ is rewarded with reward $r$. The final structure is 
\begin{equation*}
\mathcal{TD}' = \left((e_1,\dotsc, e_N, e), \mathcal{T}, \mathcal{L}: \{a: \mathcal{L}(a) + r\}, \mathcal{M} : \{e: \textrm{candidate}\}\right )    
\end{equation*}
\end{definition}

In the definition above, we use the terminology $\mathcal{L} : \{a: t\}$ to mean that ledger $\mathcal{L}$ is modified so that $\mathcal{L}(a) = t$, and terminology $\mathcal{M} : \{e : \text{candidate}\}$ to denote that $\mathcal{M}$ is modified to hold metadata denoting $e$ as a candidate.

For makers $a$ proposing the addition of a leaf node $e$ holding value $v$, they will be responsible for storing the data off-chain since other nodes have the option of issuing a \textit{challenge} that can remove the element $e$ from the dataset. If agent $a$ can't produce $v$ upon query, token holders will be incentivized to vote for removal of $e$ from $\mathcal{TD}$.

An important special case for the candidate is the candidacy stake $S_\text{candidate}$ is set to $0$. This setting will be crucial for cases when constructing $\mathcal{TD}$ will be challenging and barriers for candidacy need to be low. For example, a tokenized dataset will likely require the minimum deposit for candidacy to be zero since otherwise it will be challenging to incentivize workers to participate in the project.

All elements that pass candidacy are said to be \textit{quasi-finalized}.

\begin{definition}[Quasi-Finality]
An element $e_i$ is \textit{quasi-finalized} when it has passed candidacy and its addition to the tokenized dataset has been approved by a quorum of token holders in $\mathcal{TD}$. 
Note that quasi-finalized elements may still be challenged by token holders. Upon being quasi-finalized, the metadata associated with $e_i$ is updated.
\begin{equation*}
\mathcal{TD}' = \left((e_1,\dotsc, e_N), \mathcal{T}, \mathcal{L}, \mathcal{M} : \{e_i: \textrm{quasi-finalized}\}\right )    
\end{equation*}
\end{definition}

Note that quasi-finalized elements may still be challenged by token holders. Initially, the proposing agent $a$ still has its candidate reward $r$ bonded to the $\mathcal{TD}$. This stake can be seized via challenge. If the challenge succeeds, reward $r_i$ will be seized from agent $a$ and the element $e_i$ will be flagged in the metadata $\mathcal{M}$ as successfully challenged.

\begin{definition}[Challenge]
Let $\mathcal{TD} = \left ((e_1,\dotsc, e_N), \mathcal{T}, \mathcal{L}, \mathcal{M} \right)$ be a tokenized data structure. In a \textit{challenge}, agent $a$ proposes the modification of metadata associated with element $e_i$ from $\mathcal{TD}$ to denote that this datapoint has been challenged. Agent $a$ must place $S_{\text{challenge}}$ tokens at stake for the duration of the vote. If a quorum $t_{\text{quorum}}$ of the token holders recorded in $\mathcal{L}$ authorize the modification, element $e_i$ is marked as challenged in $\mathcal{TD}$. Let $a'$ be the agent who originally added $e_i$. Then $a'$'s reward is seized and the final structure is 
\begin{equation*}
\mathcal{TD}' = \left((e_1,\dotsc, e_N) \setminus \{e_i\}, \mathcal{T}, \mathcal{L}:\{a': \mathcal{L}(a') - r\}, \mathcal{M}:\{e_i: \text{challenged} \}\right )    
\end{equation*}
\end{definition}

When the reward for $r$ for proposing agent $a$ is no longer bonded to the network, challenges can no longer seize the reward, and the ledger is not amended upon a successful challenge, but the associated metadata still is.

\begin{definition}[Statute Of Limitations]
An element $e_i$ has exceeded its \textit{statute of limitations} when the reward $r$ issued to its proposing agent $a$ is no longer bonded to the $\mathcal{TD}$. Challenges may still be issued against $e_i$ but it will no longer be possible to seize the candidate reward $r$.
\end{definition}

\subsection{Liveness of Data}
Since the data in leaf nodes may be stored off-chain, ensuring liveness of data is critical. Note that when a leaf node with data is accepted into a TD substructure, the leaf node owner is rewarded with minted substructure tokens.

Token holders for the substructure are incentivized to challenge leaf nodes to prove they can produce their data. If the challenge passes muster, the leaf node tokens are burned, implicitly raising the value of existing token holders' holdings. Recall that token holders are entitled to a fraction of future membership payments in proportion to their fractional ownership of the token supply. Hence, a token holder is incentivized to increase their fractional ownership (and gain rights to a larger share of future returns) by pruning dead leaf nodes.

\begin{theorem}[Proof of Liveness] Let $\mathcal{TD}$ be a tokenized dataset. Suppose that the current saleable value of $\mathcal{TD}$ is $\#D$ units in a base currency and that potential for $k$ future sales exists. Assume that a liveness check costs $\#c$ units. Then, token holders of a dataset are incentivized to check any element $e \in \mathcal{TD}$ a total of $\ell = \frac{kD}{cN}$ times for data liveness.\end{theorem}
\begin{proof}
Let's suppose that $\mathcal{TD} = ((e_1,\dotsc,e_N), \mathcal{T}, \mathcal{L}, \mathcal{M})$. On average, each element $e_i$ has saleable value $\#D/N$ units. Given the potential for $k$ future sales, the future value of each data element is $\#kD/N$ units. It follows that a token holder is economically incentivized up to $\ell = \frac{kD}{cN}$ checks for data liveness.
\end{proof}

It's useful to substitute actual numbers to gain some intuition for this result. Let's suppose that $\mathcal{TD}$ has saleable value $D=\#1000$ units at present and that potential for $k=10$ future sales exists. Let's suppose that a liveness check costs $c=\#.1$ units and that the dataset has $N=10000$. Then a token holder is incentivized to issue $\ell=\frac{kD}{cN} = \frac{10*\#1000}{\#.1 * 10000} = 10$ liveness checks for \textit{each} datapoint in $\mathcal{TD}$.

\subsection{Accessing private data}

A tokenized data structure $\mathcal{TD}$ explicitly allows for the addition of private off-chain data. The introduction of this primitive raises new questions: How can an interested party honestly gain access to this off-chain data in a way that fairly rewards the token holders who have curated this resource? And more importantly, what are the new attack vectors that arise as a result of this new resource?

In this section, we will consider two potential access modes by which interested parties can access private, off-chain data. Namely, \textit{membership} and \textit{transactions}.

\subsubsection{Membership Model}

In the membership model, any interested agent who wishes to access private data must become a token holder who holds stake in $\mathcal{TD}$. Then requesting to query a private datapoint $v_i$ stored in element $e_i$ requires placing stake $S_\text{query}$ within the system. The process of acquiring $S_\text{query}$ stake in $\mathcal{TD}$ is referred to as the process of acquiring \textit{membership} in $\mathcal{TD}$.

\begin{definition}[Membership]
Agent $a$ acquires \textit{membership} in tokenized data structure $\mathcal{TD}$ by acquiring $S_\text{query}$ tokens in its token economy.
\end{definition}

How does agent $a$ acquire $S_\text{query}$ tokens? Let's assume that $S_\text{query}$ tokens holds $\#D$ units of value. The payment of $\#D$ tokens is then split out pro-rata (according to ownership share) among all present token holders.

When analyzing behavior, it will be useful to assume that the economy is in steady state, so that tokens are no longer being issued.

\begin{definition}[Steady State]
A tokenized data structure $\mathcal{TD}$ is in steady state if $\#(\mathcal{T})$ can no longer change.
\end{definition}

With the definitions of membership and steady state laid down, it becomes possible to analyze the expected rewards for honest and dishonest behavior. As before, let $\#(\mathcal{T})$ denote the number of tokens of type $\mathcal{T}$ available. We start with a useful definition.

\begin{definition}[Leakage Resistance]
Let $a$ be a token holder in tokenized data structure $\mathcal{TD}$ that has present market value $\#V$ units. Let's suppose that $a$ leaks private data from $\mathcal{TD}$ and that the post-leakage market value of $\mathcal{TD}$ is $\#\beta V$ units. Then we say that $\mathcal{TD}$ has \textit{leakage resistance} $\beta$.
\end{definition}

\begin{definition}[Counterfeit Worth]
Let $a$ be a token holder in tokenized data structure $\mathcal{TD}$ that has present market value $\#V$ units. Let's suppose that $a$ leaks private data from $\mathcal{TD}$ and that the market value of the leaked information from $\mathcal{TD}$ is $\#\gamma V$ units. Then we say that $\mathcal{TD}$ has \textit{counterfeit worth} $\gamma$.
\end{definition}

Leakage resistance $\beta$ range from $0$ to $1$. Different tokenized data structures will have different leakage resistance factors $\beta$ depending on the type of data they hold. Data for which establishing provenance is critical (perhaps for regulatory reasons as in health care) may have resistance factors close to $1.0$. Data for which provenance doesn't matter (perhaps quantitative trading datasets) may have resistance factors close to $0$.

\begin{theorem}[Rewards for Honest and Dishonest Behavior] Let $a$ be an agent who buys $S_\text{query}$ tokens from tokenized data structure $\mathcal{TD}$ in steady state for $\#D$ units of value. Let $\alpha = \frac{S_\text{query}}{\#(\mathcal{T})}$ be the fractional ownership of $\mathcal{TD}$ required for querying data. Let $\beta$ be the leakage resistance of $\mathcal{TD}$ and let $\gamma$ be the counterfeit worth. Suppose that potential for $k$ future membership sales and $\ell$ counterfeit sales exists and that a nonzero probability $p_\text{detect}$ exists for dishonest behavior to be detected by token holders in $\mathcal{TD}$. Then the expected value for honest behavior $E[a_\text{honest}] > 0$ and the expected value for dishonest behavior $E[a_\text{dishonest}] < 0$ if $\alpha < \frac{p_\text{detect} + (1-p_\text{detect})\gamma \ell}{(1-p_\text{detect})\beta k}$.
\end{theorem}
\begin{proof}
Let suppose that agent $a$ pays $\#D$ units to obtain stake $S_\text{query}$ in $\mathcal{TD}$. Then let's assume that there are $k$ sales that will happen in the future. Then the total future return that $a$ can expect is $\#\alpha kD$ units of value. Put another way, the expected reward for honest behavior $E[a_\text{honest}] = \#\alpha k D > 0$ units of value.

If $a$ is dishonest and leaks information, there are two possible outcomes. The first is that the dishonesty is caught and a challenge is issued to $a$ causing loss of stake. This outcome has value $-\#D$ units. Let's assume alternatively that $a$'s dishonest behavior is not caught. In this case, the value of data visibility will drop to $\#\beta D$ units since the leakage resistance of $\mathcal{TD}$ is $\beta$. The data will also have counterfeit value $\# \gamma D$ units. In this case, the expected return is $\#(\beta \alpha k + \gamma \ell) D$ units of value. Then the expected return for dishonest behavior is $E[a_\text{dishonest}] = -p_{\text{detect}} \#D + (1-p_\text{detect}) \#(\beta \alpha k + \gamma \ell) D$. This quantity is negative if and only if
\begin{align*}
    E[a_\text{dishonest}] = -p_{\text{detect}} \#D + (1-p_\text{detect}) \#(\beta \alpha k + \gamma \ell) \#D &< 0 \\
     -p_\text{detect} + (1 - p_\text{detect})(\beta \alpha k + \gamma \ell) &< 0 \\
     \alpha &< \frac{p_\text{detect} + (1-p_\text{detect})\gamma \ell}{(1-p_\text{detect})\beta k}
\end{align*}
\end{proof}

This result is a little curious. It indicates that while increasing the required fractional ownership $\alpha$ for querying private data increases the rewards for honest behavior, it can also create positive rewards for dishonest behavior if $\mathcal{TD}$ is leakage resistant and has low counterfeit worth. In these cases, malicious parties can freely leak while still enjoying positive returns. These results suggest that designing resilient economies for tokenized data structures may take significant research to do correctly.

\subsubsection{Transaction Model}

In this model, agents who wish to gain access to data would pay a direct fee $\#F$ units to all token holders in $\mathcal{TD}$ but would not gain ownership in the the tokenized data structure. The weakness of this model is that the expectation of future returns from the dataset now no longer constrains the behavior of purchasing agents. Without this positive reward for honest behavior, leaks will become more likely and destroy the value of $\mathcal{TD}$.

\subsection{Future Returns}

Constructing a new tokenized data structure $\mathcal{TD}$ can take a significant amount of effort. What motivates a potential contributor to put forward this effort? Simply put, the contributor will make this effort if the expected monetary reward for the effort is positive.

\begin{theorem}[Expected Future Returns]
Let us suppose that $\#E$ units of capital must be expended for agent $a$ to obtain and store element $e$. Let $\mathcal{TD}$ be a tokenized data structure in steady state. Let's suppose that in the future, a total of $k$ agents will be interested in obtaining membership in $\mathcal{TD}$ for $\#D$ units value each. Then the expected return $E[a_\text{contribution}]$ for contributing $e$ to tokenized data structure $\mathcal{TD}$ with candidate reward $S_\text{candidateReward}$ is $\#\alpha k D - \#E$ where $\alpha$ is the fractional ownership of $a$ in $\mathcal{TD}$.
\end{theorem}
\begin{proof}
The fractional ownership that $a$ receives in $\mathcal{TD}$ is $\alpha = \frac{S_\text{candidateReward}}{\#(\mathcal{T})}$. Then the expected future return that $a$ will receive for its work is $\#\alpha k D$ leading to expected return $\# \alpha k D - E$.
\end{proof}

This theorem provides a corollary that guides how high $D$ must be priced for contributions to be encouraged.

\begin{corollary}[Data Pricing]
Let us suppose that $\# E$ units of capital must be expended for any agent to obtain an element $e$ suitable for tokenized data structure $\mathcal{TD}$. Then price $\#D$ for the dataset must be set greater than $\frac{E\#(\mathcal{T})}{kS_\text{candidateReward}}$ for $a$ to have positive expected return on its contribution.
\end{corollary}
\begin{proof}
\begin{align*}
    E[a_\text{contribution}] &> 0 \\
    \alpha k D - E &> 0 \\
    \frac{S_\text{candidateReward} k D}{\#(\mathcal{T})} - E &> 0 \\
    D &> \frac{E\#(\mathcal{T})}{kS_\text{candidateReward}} \\
\end{align*}
\end{proof}

Note that pricing depends on whether $\#(\mathcal{T})$ is a static or dynamic quantity. For inflationary token economics, where $\#(\mathcal{T})$ grows larger with time, the required price for $D$ will grow larger with time as well.

\subsection{Forking a $\mathcal{TD}$}

\begin{definition}[Forks]
A fork is an operation which proposes splitting a given $\mathcal{TD}$ into two separate $\mathcal{TD}_1$ and $\mathcal{TD}_2$ structures. The elements $((e_1,\dotsc,e_N), \mathcal{T}, \mathcal{L}, \mathcal{M})$ must be divided (without overlap) between the two children structures. This means in particular that the two child ledgers cannot intersect $\mathcal{L}_1 \cap \mathcal{L}_2 = \emptyset$
\end{definition}

An agent who wishes to trigger a fork must place $S_\text{forkDeposit}$ at stake. This deposit triggers a forking period. All tokens holders on a registry have time $t_\text{forkVote}$ to adopt one of the two forked registries.
Adoption of one registry means token holdings on the other registry are destroyed.

The recursive nature of tokenized data structures introduces an interesting complicated though. Let's suppose that a given $\mathcal{TD}$ contains element $e_1$ which is itself a tokenized data structure $\mathcal{TD}_1$. Let's say that a fork is triggered for $\mathcal{TD}_1$ which splits the data structure into $\mathcal{TD}_{11}$ and $\mathcal{TD}_{12}$. By convention, let's agree that $\mathcal{TD}_{11}$ is the direct offshoot and $\mathcal{TD}_{12}$ is the forked variant. Note then that $\mathcal{TD}_{12}$ is not yet an element of $\mathcal{TD}$! The token holders in $\mathcal{TD}_{12}$ will need to apply for candidacy in $\mathcal{TD}$. This extra candidacy step place an additional hurdle to discourage frivolous forks.

\subsection{Slashing Conditions}

Slashing conditions for tokenized data structures are implicitly implemented via the challenge mechanism. This implicit scheme can be significantly more robust than an automated slashing condition since there need not be a simple algorithmic rule for slashing. Human (or intelligent agent) token holders can issue challenges for arbitrary reasons including suspicion of fraud that is hard to prove with a rigid algorithmic condition.

\subsection{Token Valuations}

A tokenized data structure $\mathcal{TD}$ depends critically on its associated token $\mathcal{T}$. What is the economic value of such a token in equilibrium conditions? We have discussed the economic rewards that accrue to token holders at depth already. In this model, the value of the token is directly proportional to the future economic rewards that will accrue to a token holder. Let's suppose that a total of $\#R$ units of discounted future economic rewards will accrue to $\mathcal{TD}$. Then the unit value of a token should be $\#R/\#(\mathcal{T})$. This simple heuristic provides justification for why a data structure token has value, but more refined analysis is left to future work.

\section{Attacks}

In this section, we consider a number of possible attacks upon tokenized data structures. We discuss the severity of each form of attack and consider potential mitigation strategies.

\subsection{Dilution of Token Economies with Depth}
The recursive definition of a tokenized data structure means that nesting can go arbitrarily deep. As a result, tokens generated for nested substructures of $\mathcal{TD}$ can have very small economies. For this reason, such substructures may be especially prone to other attacks. This is related to the minimum economy size problems identified in the TCR paper \cite{goldin2017tcr}, which suggests that TCRs may not be well suited for small economy problems such as generating grocery lists. In the case of a decentralized data exchange, smaller exchanges that are deeply nested within the broader tokenized data market may not have sufficient economic protections to discourage attacks.

\subsection{Trolling Attacks}
The TCR paper \cite{goldin2017tcr} identifies trolling attacks as a class of vulnerabilities. In this attack, trolls are actors who are willing to attack a system individually to poison gathered data. Such trolls seek to maximize chaos, but are usually not willing to suffer large personal losses to do so. The requirement for a stake to propose TCR candidates means that trolling attacks should be relatively ineffective since the loss of stake resulting from challenes could make such attacks expensive.

For tokenized data structures, it is possible that trolling attacks could prove more dangerous. As we have discussed above, it may make sense to set $S_\text{candidateDeposit}$ to $0$ in order to lower barriers for potential contributors to $\mathcal{TD}$. In such a case, the economic barriers against trolling attacks no longer pose a barrier. However, if token holders can algorithmically detect troll-submitted candidates with low effort, then the severity of such attacks may be mitigated.

\subsection{Madman Attacks}
The TCR paper \cite{goldin2017tcr} introduces madman attacks, where motivated adversaries are willing to undergo economic losses to poison a registry. For example, a corporation or nation-state may seek to thwart the construction of particular data structure which could hurt its interests. Such adversaries may be willing to pay large sums to thwart the construction of such structures.

Defenses against madman attacks are limited by the size of the economy. For tokenized data structures with large economies, such attacks will be prohibitively expensive, but for smaller economies, these attacks will likely prove damaging. These attacks could prove challenging for decentralized data exchanges, where existing data brokers could be motivated to attack decentralized datasets that challenge their market position. Future work needs to consider how to mitigate such attacks.

\subsection{Sybil Attacks}

Token holders can use multiple coordinated accounts to game a tokenized data structure. In this section, we discuss a few possible such attacks and their effects.

\subsubsection{Data duplication}
A token holder can propose the addition of data element $e$ to $\mathcal{TD}$. If the data element is valid, this will merit a reward $r$ issued to the proposer. Once the reward bonding period is complete, the token holder can use a second account to challenge the liveness of $e$ and purposefully fail the liveness challenge. At this point, the metatdata for $e$ would be modified to note that it has been challenged. The proposer can use a new account to propose re-adding $e$ to $\mathcal{TD}$. Performed iteratively, this scheme could repeatedly gain rewards for the same datapoint.

To defend against the attack, the creator of $\mathcal{TD}$ can choose to make the reward bonding period large. Alternatively, parameters could be set so refreshing a previously challenged element might earn a much smaller reward. This modification could close off the duplication attack, but might lower incentives for agents to refresh dead data. 

It's also likely that many tokenized data structure creators will require their makers to provide proof of their identity. Known identities will create another layer of accountability that will make duplication attacks more challenging.

\subsection{Data Leakage}

The data stored off-chain on $\mathcal{TD}$ will likely leak over time as the number of agents who have accessed the dataset increase. We have argued in the membership model that purposeful leakage is economically disincentivized, but it's likely that residual leakage will happen over time. It remains an open problem to construct a membership model that will minimize leakage over long time periods.

\subsection{Forking Attacks}

A malicious agent could seek to trigger adversarial forks of $\mathcal{TD}$ in order to gain additional control. However, frivolous forks will likely not gain broad backing from the token holders of $\mathcal{TD}$ so even if a fork is triggered, the offshoot branch will have a much smaller economy. This will limit the potential economic gain for the malicious agent.

\section{Discussion}

Decentralized data markets might prove very useful for the development of intelligent agents. For example, a deep reinforcement learning \cite{mnih2015human} or evolutionary agent with a budget could algorithmically construct a tokenized dataset $\mathcal{TD}$ to solicit the construction of a dataset needed to further train the existing model. Significant research progress in deep reinforcement learning has resulted in the design of agents that can learn multiple sets of skills \cite{mankowitz2018unicorn}, so it seems feasible for an agent to learn how to budget dataset gathering requests. More prosaically, data scientists and researchers can access the decentralized data exchanges to find datasets thay may prove useful for their work. Access to data liquidity could prove a powerful tool for democratization of machine learning and AI models

In tokenized schemes, it's common to ask whether the token is a necessary part of the design. Couldn't the same design be constructed using an existing token such as ETH or BTC? It does in fact seem likely that a tokenized data structure can be meaningfully constructed with all stakes placed in ETH for example. However, it's not possible to issue recursive sub-tokens if we insist on ETH stakes; requiring contributors to a $\mathcal{TD}$ to front significant capital makes it unlikely that they will participate. For this reason, we suspect that tokenized data structures without custom tokens will face major challenges constructing nontrivial data structures.

In the present work, we have limited our analysis to a mathematical presentation of the properties of tokenized data structures. We leave for future work the nontrivial challenge of implementing tokenized data structures on an existing smart contract platform such as Ethereum \cite{buterin2013ethereum}.

\section{Conclusion}

In this work, we demonstrate how to construct a decentralized data exchange. This construction is built upon the primitive of tokenized data structures. Such tokenized data structures combine the strengths of past work on token curated registries \cite{goldin2017tcr} and distributed hash tables \cite{stoica2003chord} to provide a framework for constructing incentivized data structures capable of holding off-chain, private data. In addition, tokenized data structures introduce the notion of recursive sub-tokens to incentivize contributors. We provide a mathematical framework for analyzing tokenized data structures and prove theorems that show that participants in a tokenized data structure $\mathcal{TD}$ are incentivized to construct the data structure for positive expected rewards. We discuss how these incentives allow for the construction of robust decentralized data markets. We conclude by discussing how such decentralized data markets could prove useful for the future development of machine learning and AI.

It's worth noting that our theorems don't prove Byzantine Fault Tolerance of tokenized data structures against adversaries. Rather they provide much weaker guarantees that honest participants will benefit from participating. We provide qualitative arguments why tokenized data structures are robust against some classes of adversarial attacks, but a more rigorous formal treatment is left to future work.

\medskip
\printbibliography
\end{document}